\documentclass[12pt]{article}

\usepackage{amssymb,url}
\usepackage{graphicx,fullpage}
\usepackage{listings}

\newtheorem{theorem}{Theorem}
\newtheorem{lemma}{Lemma}

\newenvironment{proof}[1][Proof]{\begin{trivlist}
\item[\hskip \labelsep {\bfseries #1}]}{\end{trivlist}}

\def\H{\mathcal{H}}
\def\KL{\mathrm{KL}}
\def\X{\mathcal{X}}
\def\C{\mathcal{C}}

\def\tx{\tilde{x}}
\def\tq{\tilde{q}}
\def\tp{\tilde{p}}
\def\ta{\tilde{a}}
\def\tg{\tilde{g}}
\def\tc{\tilde{c}}
\def\th{\tilde{h}}
\def\tH{\tilde{H}}

\author{Frank Nielsen\thanks{A first version of this paper appeared in IEEE Signal Processing Letters~\protect\cite{JeffreysCentroidSPL-2013}. This revision includes a faster fixed-point technique in Section~\ref{sec:fpi} and the R source code. \protect\url{www.informationgeometry.org} (5793b870)}\\
Sony Computer Science Laboratories, Inc.\\
3-14-13 Higashi Gotanda, 141-0022 Shinagawa-ku, Tokyo, Japan 
}

\title{On the symmetrical Kullback-Leibler Jeffreys centroids }
\date{}

\begin{document}

\maketitle

\begin{abstract}
Due to the success of the bag-of-word modeling paradigm, clustering histograms has become an important ingredient of modern information processing.
Clustering histograms can be performed using the celebrated $k$-means centroid-based algorithm.
From the viewpoint of applications, it is usually required to deal with symmetric distances. 
In this letter, we consider the Jeffreys divergence that symmetrizes the Kullback-Leibler divergence, and investigate the computation of  Jeffreys centroids.
We first prove that the Jeffreys centroid can be expressed {\em analytically}  using the Lambert $W$ function for {\em positive} histograms.
We then show how to obtain a fast guaranteed approximation when dealing with {\em frequency} histograms.
Finally, we conclude with some remarks on the $k$-means histogram clustering.
\end{abstract}

%%%%%%%%%%%%%%%%%%%%%%
\section{Introduction: Motivation and prior work}
%%%%%%%%%%%%%%%%%%%%%%

\subsection{Motivation: The Bag-of-Word modeling paradigm}

{Classifying} documents into categories is a common task of information retrieval systems:
Given a training set of  documents labeled with categories, one asks to classify incoming new documents.
Text categorization~\cite{TextCategorization-2003} proceeds by  first defining a dictionary of words (the corpus).
It then models each document by a {\em word count} yielding a  word histogram per document.
Defining a proper distance $d(\cdot,\cdot)$ between histograms allows one to:

\begin{itemize}

\item Classify a new on-line document: we first calculate its histogram signature and then seek for the labeled document which has the {\em most similar} histogram to deduce its tag ({\it e.g.}, using a nearest neighbor classifier). 

\item Find the initial set of categories: we  cluster all document histograms and assign a category per cluster.

\end{itemize}

It has been shown experimentally that the  Jeffreys divergence (symmetrizing the Kullback-Leibler divergence) achieves better performance than the traditional {\em tf-idf} method~\cite{TextCategorization-2003}.
This text classification method based on the Bag of Words (BoWs) representation has also been instrumental in computer vision for
efficient object categorization~\cite{bok-2004} and recognition in natural images. 
It first requires to create a dictionary of ``visual words'' by quantizing keypoints   of the training database. 
Quantization is then performed using the $k$-means algorithm~\cite{bregmankmeans-2005} that partitions $n$ data points $\X=\{x_1, ..., x_n\}$ into $k$ clusters $\C_1, ..., \C_k$ where each data element belongs to the closest cluster center. 
From a given initialization, Lloyd's batched $k$-means first assigns points to their closest cluster centers, and then
update the cluster centers, and reiterate this process until convergence is met after a finite number of steps.
When the distance function $d(x,y)$ is chosen as the squared Euclidean distance $d(x,y)=\|x-y\|^2$, the cluster centroids updates to their   centers of mass. 
Csurka et al.~\cite{bok-2004} used the squared Euclidean distance for building the visual vocabulary.
Depending on the considered features, other distances have proven useful:
For example, the Jeffreys divergence was shown to perform experimentally better than the Euclidean or squared Euclidean distances for  Compressed Histogram of Gradient descriptors~\cite{SKL-CHOG-2012}.
To summarize, $k$-means histogram clustering with respect to the Jeffreys divergence can be used to both quantize visual words to create a dictionary and to cluster document words for assigning initial categories.

Let $w_h=\sum_{i=1}^d h^i$ denote the cumulative sum of the bin values of histogram $h$.
We distinguish between {\em positive histograms} and {\em frequency histograms}.
A frequency histogram $\tilde{h}$ is a unit histogram ({\it i.e.}, the cumulative sum $w_{\tilde{h}}$ of its bins adds up to one).
In statistics, those positive and frequency histograms correspond respectively to {\em positive discrete} and {\em multinomial} distributions when all bins are non-empty.
Let $\H=\{h_1, ..., h_n\}$ be a collection of $n$ histograms with $d$ positive-valued bins.
By notational convention, we use the superscript and the subscript to indicate the bin number and the histogram number, respectively.
Without loss of generality, we assume that all bins are non-empty\footnote{Otherwise, we may add an arbitrary small quantity $\epsilon>0$ to all bins.
When frequency histograms are required, we then re-normalize.}: $h_j^i\geq 0, 1\leq j\leq n, 1\leq i\leq d$.
To measure the distance between two such histograms $p$ and $q$, we rely on the {\em relative entropy}.
The {\em extended KL divergence}~\cite{bregmankmeans-2005} between two positive (but not necessarily normalized) histograms $p$ and $q$ is defined by $\KL(p:q) = \sum_{i=1}^d  p^i\log\frac{p^i}{q^i} + q^i - p^i$.
Observe that this information-theoretic dissimilarity measure is not symmetric nor does it satisfy the triangular inequality property of metrics. 
Let $\tilde{p}=\frac{p}{\sum_{i=1}^d p^i}$ and $\tilde{q}=\frac{q}{\sum_{i=1}^d q^i}$ denote the corresponding normalized frequency histograms.
In the remainder, the $\tilde{\ }$ denotes this normalization operator.
The extended KL divergence formula applied to normalized histograms yields the traditional KL divergence~\cite{bregmankmeans-2005}:
$\KL(\tilde{p}:\tilde{q}) = \sum_{i=1}^d  \tilde{p}^i\log\frac{\tilde{p}^i}{\tilde{q}^i}$ since $\sum_{i=1}^d \tilde{q}^i - \tilde{p}^i=\sum_{i=1}^d \tilde{q}^i-\sum_{i=1}^d \tilde{p}^i=1-1=0$.
The KL divergence is interpreted as the {\em relative entropy} between $\tilde{p}$ and $\tilde{q}$: $\KL(\tp:\tq)= H^\times(\tp:\tq)-H(\tp)$,
where $H^\times(\tp:\tq)= \sum_{i=1}^d \tilde{p}^i\log\frac{1}{\tilde{q}^i}$ denotes the {\em cross-entropy} and $H(\tp)=H^\times(\tp:\tp)=\sum_{i=1}^d \tilde{p}^i\log\frac{1}{\tilde{p}^i}$ is the {\em Shannon entropy}. This distance is explained as the expected extra number of bits per datum that must be transmitted when using the ``wrong'' distribution $\tq$ instead of the true distribution $\tp$. Often $\tp$ is hidden by nature and need to be hypothesized while $\tq$ is estimated.
When clustering histograms, all histograms play the {\em same role}, and it is therefore better to consider the
Jeffreys~\cite{JeffreysPrior-1946} divergence $J(p,q) = \KL(p:q) + \KL(q:p)$ that symmetrizes the KL divergence:

\begin{equation}\label{eq:jdiv}
J(p,q) = \sum_{i=1}^d  (p^i-q^i)\log\frac{p^i}{q^i} = J(q,p).
\end{equation}
Observe that the formula for Jeffreys divergence holds for arbitrary positive histograms (including frequency histograms).

This letter is devoted to compute efficiently the {\em Jeffreys centroid} $c$ of a set $\H=\{h_1, ...,h_n\}$ of weighted histograms defined as:
\begin{equation}\label{eq:minfrequency}
c= \arg\min_{x} \sum_{j=1}^n \pi_j J(h_j,x),
\end{equation}
where the $\pi_j$'s denote the histogram positive weights (with $\sum_{j=1}^n \pi_j=1$).
When all histograms $h_j\in\H$ are normalized, we require the minimization of $x$ to be carried out over $\Delta_{d}$, the $(d-1)$-dimensional probability simplex. This yields the {\em Jeffreys frequency centroid} $\tc$.
Otherwise, for positive histograms  $h_j\in\H$, the minimization of $x$ is done over the positive orthant $\mathbb{R}_+^d$, to get
the {\em Jeffreys positive centroid} $c$.
Since the $J$-divergence is convex in both arguments, both the Jeffreys positive and frequency centroids are unique.

%%%%%%%%%%%% 
\subsection{Prior work and contributions} 
%%%%%%%%%%%%
On one hand, clustering histograms has been studied using various distances and clustering algorithms.
Besides the classical Minkowski $\ell_p$-norm distances,  hierarchical clustering with respect to the $\chi^2$ distance has been investigated in~\cite{Chi2-1995}.  
Banerjee et al.~\cite{bregmankmeans-2005} generalized $k$-means to Bregman $k$-means thus allowing to cluster distributions of the same exponential families  with respect
to the KL divergence.
Mignotte~\cite{kMeansBhat-2008} used $k$-means with respect to the Bhattacharyya distance~\cite{2011-brbhat} on histograms of various color spaces to perform image segmentation. 
On the other hand, Jeffreys $k$-means has not been yet extensively studied as the   centroid computations are non-trivial:
In 2002, Veldhuis~\cite{Veldhuis-2002} reported an iterative Newton-like algorithm to approximate arbitrarily finely the Jeffreys frequency centroid $\tc$ of a set of frequency histograms that requires two nested loops.
Nielsen and Nock~\cite{2009-BregmanCentroids-TIT} considered the information-geometric structure of the manifold of multinomials (frequency histograms) to report a simple geodesic bisection search algorithm ({\em i.e.}, replacing the two nested loops of~\cite{Veldhuis-2002}   by one single loop). 
Indeed, the family of frequency histograms belongs to the exponential families~\cite{bregmankmeans-2005}, and the Jeffreys frequency centroid amount to compute equivalently a symmetrized Bregman centroid~\cite{2009-BregmanCentroids-TIT}.

To overcome the explicit computation of the Jeffreys centroid, Nock et al.~\cite{MixedBregmanClustering-2008} generalized the Bregman $k$-means~\cite{bregmankmeans-2005} and $k$-means++ seeding using {\it mixed Bregman divergences}:
They consider two dual centroids $\tc_m$ and $\tc^*_m$  attached per cluster, and use the following   divergence depending on these two centers:
$\Delta\KL(\tc_m:x:\tc_m^*) = \KL(\tc_m:x) + \KL(x:\tc_m^*)$.
However, note that this mixed Bregman $2$-centroid-per-cluster clustering is {\em different} from the Jeffreys $k$-means clustering that relies on one centroid per cluster.

This letter is organized as follows: 
Section~\ref{sec:unnormalized} reports a closed-form expression of the  positive Jeffreys centroid for a set of positive  histograms. 
Section~\ref{sec:normalized} studies the guaranteed tight approximation factor obtained when normalizing the positive Jeffreys centroid, and further describes a simple bisection algorithm to arbitrarily finely approximate the optimal Jeffreys frequency centroid.
Section~\ref{sec:exp} reports on our experimental results that show that our normalized approximation is in practice tight enough to avoid doing the bisection process. Finally, Section~\ref{sec:concl} concludes this work.

%%%%%%%%%
\section{Jeffreys positive centroid\label{sec:unnormalized}}
%%%%%%%%%

We consider a set $\H=\{h_1, ..., h_n\}$ of $n$ positive weighted histograms with $d$ non-empty bins ($h_j\in\mathbb{R}_+^d$, $\pi_j>0$ and $\sum_j \pi_j=1$).
The  {\em Jeffreys positive centroid} $c$ is defined by:

\begin{equation}
c =\arg\min_{x\in\mathbb{R}_+^d} J(\H,x) = \arg\min_{x\in\mathbb{R}_+^d} \sum_{j=1}^n \pi_j J(h_j,x).\label{eq:min}
\end{equation}

We state the first result:
\begin{theorem}
The  Jeffreys positive centroid $c=(c^1, ..., c^d)$ of a set $\{h_1, ..., h_n\}$ of $n$ weighted  positive histograms with $d$ bins can be calculated component-wise {\em exactly} using the Lambert $W$ analytic function:
$c^i = \frac{{a}^i}{W(\frac{{a}^i}{{g}^i} e)}$,
where ${a}^i=\sum_{j=1}^n \pi_j h_j^i$ denotes the coordinate-wise arithmetic weighted means and ${g}^i=\prod_{j=1}^n  (h_j^i)^{\pi_j}$ the coordinate-wise geometric weighted means.
\end{theorem}

\begin{proof}
We seek for $x\in\mathbb{R}_+^d$ that minimizes Eq.~\ref{eq:min}.
After expanding   Jeffreys divergence formula of Eq.~\ref{eq:jdiv} in Eq.~\ref{eq:min} and removing all additive terms independent of $x$, we find
the following equivalent minimization problem:  
$$
\min_{x\in\mathbb{R}_+^d} \sum_{i=1}^d  x^i\log\frac{x^i}{g^i}  -a^i\log x^i.
$$
This  optimization can be performed coordinate-wise, independently.
For each coordinate, dropping the superscript notation and setting the derivative to zero, we have to solve 
$\log \frac{x}{{g}}+1-\frac{{a}}{x}   =0$,
which yields $x=\frac{{a}}{W(\frac{{a}}{{g}} e)}$, where $W(\cdot)$ denotes the Lambert $W$ function~\cite{LambertWfunction-1995}.
\end{proof}

Lambert function\footnote{We consider only the branch $W_0$~\cite{LambertWfunction-1995}   since arguments of the function are always positive.} $W$ is defined by $W(x)e^{W(x)}=x$ for $x\geq 0$. That is, the Lambert function is   the functional inverse of $f(x)=x e^x=y$: $x=W(y)$.
Although function $W$ may seem non-trivial at first sight, it is a popular elementary analytic function similar to the logarithm or exponential functions.
In practice,  we get a fourth-order convergence algorithm to estimate it by implementing Halley's numerical root-finding method.
It requires fewer than $5$ iterations to reach machine accuracy using the IEEE 754 floating point standard~\cite{LambertWfunction-1995}.
Notice that the Lambert $W$ function plays a particular role in information theory~\cite{LambertW}.

%%%%%
\section{Jeffreys frequency centroid\label{sec:normalized}}
%%%%%
We consider a set $\tH$ of $n$ frequency histograms: $\tH=\{\th_1, ..., \th_n\}$.

\subsection{A guaranteed approximation}

If we relax $x$ to the positive orthant  $\mathbb{R}_+^d$ instead of the probability simplex, we get the optimal positive Jeffreys centroid $c$, with 
$c^i = \frac{{a}^i}{W(\frac{{a}^i}{{g^i}} e)}$ (Theorem~1).
Normalizing this positive Jeffreys centroid to get $\tc'=\frac{c}{w_c}$ {\em does not} yield the Jeffreys frequency centroid $\tc$ that requires dedicated iterative optimization algorithms~\cite{Veldhuis-2002,2009-BregmanCentroids-TIT}. 
In this section, we consider approximations of the Jeffreys frequency histograms.
Veldhuis~\cite{Veldhuis-2002} approximated the Jeffreys frequency centroid $\tc$ by
$\tilde{c}'' = \frac{\ta+\tg}{2}$,   
where $\tilde{a}$ and $\tilde{g}$ denotes the normalized weighted arithmetic and geometric  means, respectively. 
The normalized geometric weighted mean $\tg=(\tg^1, ..., \tg^d)$ is defined by $\tilde{g}^i =  \frac{\prod_{j=1}^n (\th_j^i)^{\pi_j}}{\sum_{i=1}^d \prod_{j=1}^n (\th_j^i)^{\pi_j}}, i\in\{1, ..., d\}$.
Since $\sum_{i=1}^d \sum_{j=1}^n \pi_j \th_j^i= \sum_{j=1}^n \pi_j  \sum_{i=1}^d \th_j^i = \sum_{j=1}^n \pi_j= 1$, the normalized arithmetic weighted mean has coordinates: $\tilde{a}^i=\sum_{j=1}^n \pi_j \th_j^i$.

We consider approximating the Jeffreys frequency centroid by normalizing the Jeffreys positive centroid $c$: $\tc'=\frac{c}{w_c}$.

We start with a simple lemma:

\begin{lemma}
The cumulative sum $w_c$ of the bin values of the Jeffreys positive centroid $c$ of a set of frequency histograms is less or equal to one: $0<w_c\leq 1$.
\end{lemma}

\begin{proof}
Consider the frequency histograms $\tH$ as positive histograms.
It follows from Theorem~1 that the Jeffreys positive centroid $c$ is such that $w_c= \sum_{i=1}^d c^i = \sum_{i=1}^d\frac{a^i}{W(\frac{a^i}{g^i} e)}$.
Now, the arithmetic-geometric mean inequality states that $a^i\geq g^i$ where $a^i$ and $g^i$ denotes the
coordinates of the arithmetic and geometric positive means. Therefore $W(\frac{a^i}{g^i} e)\geq 1$ and $c^i\leq a^i$.
Thus $w_c= \sum_{i=1}^d c^i\leq  \sum_{i=1}^d a^i=1$.
\end{proof}
 
We consider approximating Jeffreys frequency centroid on the probability simplex $\Delta_d$ by using the normalization of the  Jeffreys positive centroid:
$\tilde{c}' =  \frac{a^i}{w W(\frac{a^i}{g^i} e)}$,
with $w=\sum_{i=1}^d \frac{a^i}{W(\frac{a^i}{{g}^i} e)}$.
To study the quality of this approximation, we use the following lemma: 

\begin{lemma}
For any histogram $x$ and frequency histogram $\th$, we have
$J(x,\th) = J(\tx,\th) + (w_x-1) (\KL(\tx:\th)+\log w_x)$,
where $w_x$ denotes the normalization factor ($w_x=\sum_{i=1}^d x^i$).
\end{lemma}

\begin{proof}
It follows from the definition of Jeffreys divergence and the fact that $x^i=w_x \tx^i$ that 
$J(x,\th) =  \sum_{i=1}^d (w_x\tx^i-\th^i)\log \frac{w_x \tx^i}{\th^i}$.
Expanding and mathematically rewriting the rhs. yields $J(x,\th) =   \sum_{i=1}^d  ( w_x\tx^i\log\frac{\tx^i}{\th^i}+w_x\tx^i\log w_x + \th^i\log \frac{\th^i}{\tx^i}   -\th^i\log w_x )
= (w_x-1)\log w_x + J(\tx,\th) + (w_x-1)\sum_{i=1}^d \tx^i\log\frac{\tx^i}{\th^i}=
J(\tx,\th) + (w_x-1) (\KL(\tx:\th) +\log w_x)$,
since $\sum_{i=1}^d \th^i=\sum_{i=1}^d \tx^i=1$.
\end{proof}

The lemma can be extended to a set of weighted frequency histograms $\tilde{\H}$:
$$
J(x,\tH) = J(\tx,\tH) + (w_x-1) (\KL(\tx:\tH)+\log w_x),
$$
where $J(x,\tH)=\sum_{j=1}^n \pi_j J(x,\th_j)$ and $\KL(\tx:\tH)=\sum_{j=1}^n \pi_j \KL(\tx,\th_j)$  (with $\sum_{j=1}^n \pi_j=1$).

We state the second theorem concerning our guaranteed approximation:

\begin{theorem}
Let $\tc$ denote the Jeffreys frequency centroid and $\tc'=\frac{c}{w_c}$ the normalized Jeffreys positive centroid.
Then the approximation factor $\alpha_{\tc'}=
\frac{J(\tc',\tH)}{J(\tc,\tH)}$ is such that
$1\leq  \alpha_{\tc'} \leq 1+(\frac{1}{w_c}-1)\frac{\KL(c,\tH)}{J(c,\tH)} \leq \frac{1}{w_c}$ (with $w_c\leq 1$).
\end{theorem}

\begin{proof}
We have $J(c,\tH)\leq J(\tc,\tH)\leq J(\tc',\tH)$. 
Using Lemma~2, since $J(\tc',\tH)=J(c,\tH)+(1-w_c)(\KL(\tc',\tH)+\log w_c))$ and $J(c,\tH)\leq J(\tc,\tH)$, it follows that
$1\leq \alpha_{\tc'}\leq 1+\frac{(1-w_c)(\KL(\tc',\tH)+\log w_c)}{J(\tc,\tH)}$.
We also have $\KL(\tc':\tH)=\frac{1}{w_c} \KL(c,\tH)-\log w_c$ (by expanding the KL expression and using the fact that $w_c=\sum_i c^i$).
Therefore $\alpha_{\tc'}\leq 1+\frac{(1-w_c)\KL(c,\tH)}{w_c J(\tc,\tH)} $.
Since  $J(\tc,\tH)\geq J(c,\tH)$ and $\KL(c,\tH)\leq J(c,\tH)$, we finally obtain $\alpha_{\tc'}\leq \frac{1}{w_c}$.
\end{proof}
When $w_c=1$ the bound is tight.
Experimental results described in the next section shows that this normalized Jeffreys positive centroid $\tc'$ almost coincide with the Jeffreys frequency centroid.

\subsection{Arbitrary fine approximation by bisection search}

It has been shown in~\cite{Veldhuis-2002,2009-BregmanCentroids-TIT} that minimizing Eq.~\ref{eq:minfrequency} over the probability simplex $\Delta_d$ amounts to minimize the following equivalent problem:
\begin{equation}\label{eq:min2}
\tc = \arg\min_{\tx\in\Delta_d} \KL(\ta:\tx)+\KL(\tx:\tg),
\end{equation}
Nevertheless, instead of using the two-nested loops of Veldhuis' Newton scheme~\cite{Veldhuis-2002}, we can design a single loop optimization algorithm.
We consider the Lagrangian function obtained by enforcing the normalization constraint $\sum_i c^i=1$ similar to~\cite{Veldhuis-2002}.
For each coordinate, setting the derivative with respect to $\tilde{c}^i$, we get $\log \frac{\tilde{c}^i}{\tilde{g}^i}+1-\frac{\tilde{a}^i}{\tilde{c}^i} +\lambda =0$,
which solves as $\tilde{c}^i = \frac{\tilde{a}^i}{W\left( \frac{\tilde{a}^i e^{\lambda+1}}{\tilde{g}^i} \right)}$.
By multiplying these $d$ constraints with $\tilde{c}^i$ respectively and summing up, we deduce that  
$\lambda = -\KL(\tilde{c} : \tilde{g} ) \leq 0$ (also noticed by~\cite{Veldhuis-2002}).
From the constraints that all $c_i$'s should be  less than one, we bound $\lambda$ as follows:
$c^i = \frac{ \tilde{a}^i}{ W\left( \frac{\tilde{a}^i e^{\lambda+1}}{\tilde{g}^i} \right)} \leq 1$,
which solves for  equality when $\lambda=\log(e^{\tilde{a}^i} \tilde{g}^i) -1$.
Thus we seek for $\lambda\in [\max_i{\log(e^{\tilde{a}^i} \tilde{g}^i) -1} ,0]$.
Since $s=\sum_i c^i=1$, we have the following cumulative sum equation depending on the unknown parameter $\lambda$:
$s(\lambda) = \sum_i c^i(\lambda) = \sum_{i=1}^d \frac{\tilde{a}^i}{W\left( \frac{\tilde{a}^i e^{\lambda+1}}{\tilde{g}^i} \right)}$.
This is a monotonously decreasing function with $s(0)\leq 1$. We can thus perform a simple bisection search to  approximate the optimal value of $\lambda$, and therefore  deduce an arbitrary fine approximation of the Jeffreys frequency centroid.

\begin{table*}[htbp]

{\small
$$
\begin{array}{|l||l|l|l|l|l|}\hline
&   \alpha_c (\mathrm{opt.\ positive})
&  \alpha_{\tc'} (\mathrm{n'lized\ approx.}) 
&  w_c\leq 1 (\mathrm{n'lizing\ coeff.})
& \alpha_{\tc''} (\textrm{Veldhuis})  \\ \hline\hline
\mathrm{avg} &  0.9648680345638155 & {\bf 1.0002205080964255} & 0.9338228644308926 & 1.065590178484613 \\ \hline
\mathrm{min} &  0.906414219584823 & {\bf 1.0000005079528809} & 0.8342819488534723  & 1.0027707382095195 \\ \hline
\mathrm{max} &  0.9956399220678585 & {\bf 1.0000031489541772} & 0.9931975105809021 & 1.3582296675397754  \\ \hline
\end{array}
$$
}

\caption{Experimental performance ratio and statistics for the $30000+$ images of the Caltech-256 database.
Observe that  $\alpha_c=\frac{J(\H,c)}{J(\H,\tilde{c})}\leq 1$ since the positive Jeffreys centroid (available in closed-form) minimizes the average Jeffreys divergence criterion. Our guaranteed normalized approximation  $\tc'$ is almost optimal.
Veldhuis' simple half normalized arithmetic-geometric approximation performs on average with a $6.56\%$ error but can be far from the optimal in the worst-case ($35.8\%$).
\label{tab}
}

\end{table*}

%%%%%%%%%%%%
\subsection{A faster fixed-point iteration algorithm\label{sec:fpi}}
%%%%%%%%%%%%
Instead of performing the dichotomic search on $\lambda$ that yields an approximation scheme with {\em linear convergence}, 
we rather consider the fixed-point equation:

\begin{eqnarray} 
\lambda^* &=& -\KL(\tc(\lambda^*):\tg),\label{eq:fpeq}\\
\tc^i(\lambda^*) &=& \frac{\tilde{a}^i}{W\left( \frac{\tilde{a}^i e^{\lambda+1}}{\tilde{g}^i} \right)}, \forall i\in\{1, ..., d\}.
\end{eqnarray}
Notice that Eq.~\ref{eq:fpeq} is a fixed-point equation $x=g(x)$ with $g(x)=-\KL(\tc(x):\tg)$ in $\lambda$ from which the Jeffreys frequency centroid  is recovered: $\tc=\tc(\lambda^*)$.
Therefore, we consider a fixed-point iteration scheme:
Initially, let  $\tc_0=\ta$ (the arithmetic mean) and $\lambda_0=-\KL(\tc_0:\tg)$.
Then we iteratively update $\tc$ and $\lambda$ for $l=1, 2, ...$ as follows:

\begin{eqnarray} 
\tc^i_l &=& \frac{\tilde{a}^i}{W\left( \frac{\tilde{a}^i e^{\lambda_{l-1}+1}}{\tilde{g}^i} \right)},\\
\lambda_{l}&=& -\KL(\tc_l:\tg),
\end{eqnarray}

and reiterate until convergence, up to machine precision, is reached.
We experimentally observed faster convergence than the dichotomic search (see next section). 

In general, existence, uniqueness and convergence rate of fixed-point iteration schemes $x=g(x)$ are well-studied (see~\cite{fpt-2003}, Banach contraction principle, Brouwer's fixed point theorem, etc.).

%$g(x)=-\KL(\tc(x):\tg)$.

%%%%%%%%%
\section{Experimental results and discussion\label{sec:exp}}
%%%%%%%%%

\begin{figure}
\centering
\begin{tabular}{ccc}
\includegraphics[bb=0 0 512 512, width=0.25\textwidth]{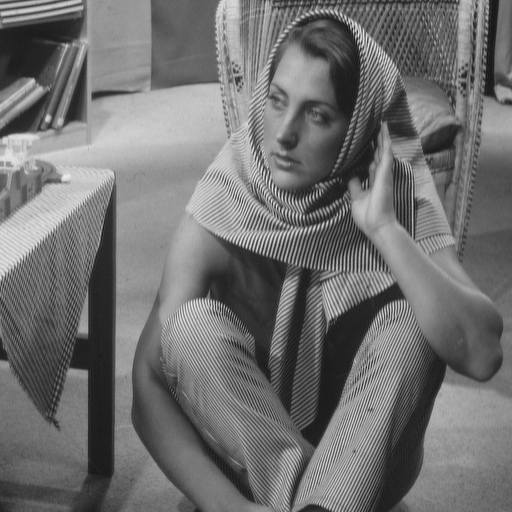} &
\includegraphics[bb=0 0 512 512, width=0.25\textwidth]{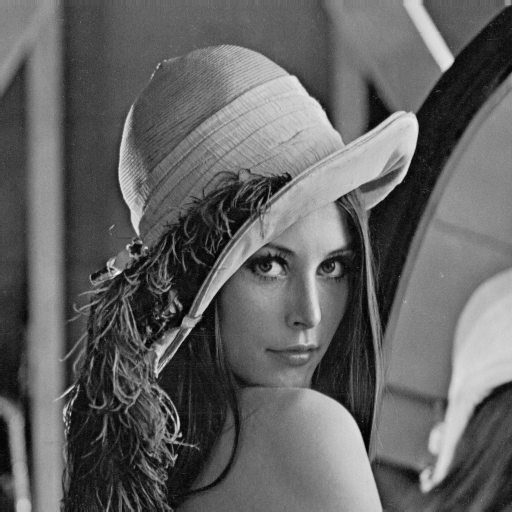} &
\includegraphics[bb=0 0 502 505, width=0.25\textwidth]{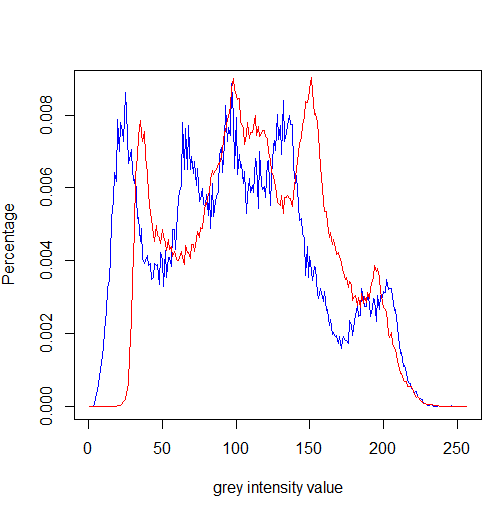} \\ %visualizehistosource.png
(a) & (b) & (c)
\end{tabular}
(d)\includegraphics[bb=0 0 592 562, width=0.5\textwidth]{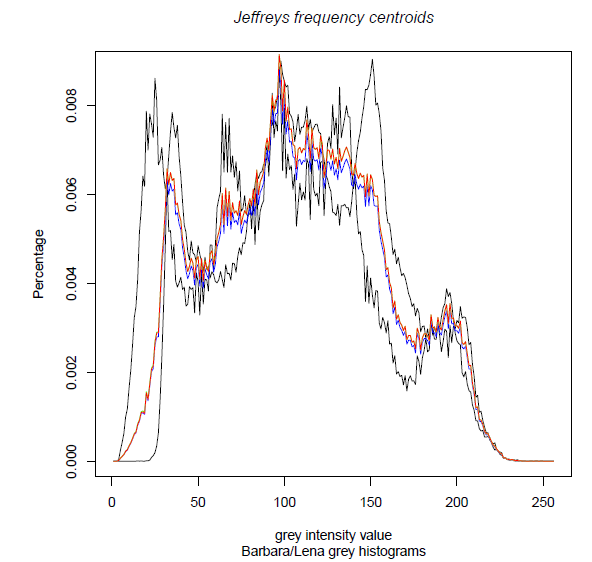}  %\includegraphics[width=0.5\textwidth]{SKLJeffreys-BarbaraLena.pdf}

\caption{Two grey-valued images: (a) {\tt Barbara} and (b) {\tt Lena}, with their frequency intensity histograms (c) in red and blue, respectively.
(d) The exact positive (blue) and frequency approximated Jeffreys (red) centroids of those two histograms.
\label{fig:twohisto}}
\end{figure}

Figure~\ref{fig:twohisto} displays the positive and frequency histogram centroids of two renown image histograms.
In order to carry out quantitative experiments, we used a multi-precision floating point (\url{http://www.apfloat.org/}) package to 
handle calculations and control arbitrary precisions. A R code is also provided in the Appendix.
We chose the Caltech-256 database~\cite{Caltech256} consisting of $30607$ images labeled into $256$ categories to perform experiments:
We consider the set of intensity\footnote{Converting RGB color pixels to $0.3R+0.596G+0.11B$ I grey pixels.} histograms $\mathcal{H}$.
For each of the $256$ category, we consider the set of histograms falling inside this category and compute the exact Jeffreys positive centroids $c$, its normalization Jeffreys approximation $\tilde{c}'$ and optimal frequency centroids $\tilde{c}$. 
We also consider the  average  of the arithmetic and geometric normalized means $\tilde{c}''=\frac{\ta+\tg}{2}$.
We evaluate the average, minimum and maximum ratio $\alpha_x=\frac{J(\H,x)}{J(\H,\tilde{c})}$ for $x\in\{c,\tc',\tc''\}$.
The results are reported in Table~\ref{tab}.
Furthermore, to study the best/worst/average performance of the the normalized Jeffreys positive centroid $\tilde{c}'$, we ran $10^6$ trials as follows:
We draw two random binary histograms ($d=2$), calculate a fine precision approximation of $\tilde{c}$ using numerical optimization, and calculate the
approximation obtained by using the normalized closed-form centroid $\tilde{c}'$. We gather statistics on the ratio $\alpha=\frac{J(\tilde{c}')}{J(\tilde{c})}\geq 1$.
We find experimentally the following performance:
$\bar\alpha\sim 1.0000009, \alpha_{\mathrm{max}}\sim 1.00181506, \alpha_{\min}=1.000000$.
Although $\tc'$ is almost matching $\tc$ in those two real-world and synthetic experiments, it remains open to express analytically and exactly its worst-case performance.

We compare experimentally the convergence rate of the dichotomic search with the fixed-point iterative scheme:
In the R implementation reported in the Appendix, we set
the precision to \verb|.Machine$double.xmin| (about $2.2250738585072013828 10^{-308}$).  
The double precision floating point number has $52$ binary digits: This coincides with the number of iterations in the bisection search method.
We observe experimentally that the fixed-point iteration scheme requires on average $5$ to $7$ iterations.

%%%%%%%%%%%%%%%%%%%%%%
\section{Conclusion\label{sec:concl}}
%%%%%%%%%%%%%%%%%%%%%%

We summarize the two  main contributions of this paper:
(1) we proved that the   Jeffreys positive centroid admits a closed-form formula expressed using the Lambert $W$ function, and
(2) we proved that normalizing this Jeffreys positive centroid yields a tight guaranteed approximation to the Jeffreys frequency centroid.
We noticed experimentally  that the closed-form normalized Jeffreys positive centroid almost coincide with the Jeffreys frequency centroid, and can therefore be used in Jeffreys $k$-means clustering.
Notice that since the $k$-means assignment/relocate algorithm
monotonically converges even if instead of computing the exact cluster centroids we update
it with provably better centroids (i.e., by applying one bisection iteration of Jeffreys frequency
centroid computation or one iteration of the fixed-point algorithm), we end up with a converging  {\em variational} Jeffreys frequency $k$-means that requires to implement a stopping criterion.  
Jeffreys divergence is not the only way to symmetrize the Kullback-Leibler divergence. 
Other KL symmetrizations include the Jensen-Shannon divergence~\cite{Jensen-Shannon-divergence},  the Chernoff divergence~\cite{Chernoff-1952}, 
and a smooth family of symmetric divergences including the Jensen-Shannon and Jeffreys divergences~\cite{symJensen:2010}.

%A Java\texttrademark{} source code is available at\\ \centerline{\url{www.informationgeometry.org/JeffreysCentroid/}}

\bibliographystyle{plain}

\section{R code}

\begin{lstlisting}
# (C) 2013 Frank Nielsen
# R code, Jeffreys positive (exact) , Jeffreys frequency approximation (bisection and fixed-point methods)
options(digits=22)
PRECISION=.Machine$double.xmin

# replace the folder name below by the folder where you put the histogram files
setwd("C:/testR")
lenah=read.table("Lena.histogram")
barbarah=read.table("Barbara.histogram")

LambertW0<-function(z)
	{
		S = 0.0
		for (n in 1:100)
		{
			 Se = S * exp(S)
			 S1e = Se + exp(S)  
			if (PRECISION > abs((z-Se)/S1e)) {break}
			S = S -	(Se-z) / (S1e - (S+2) * (Se-z) / (2*S+2))
		}
		S
	}
	
#
# Symmetrized Kullback-Leibler divergence (same expression for positive arrays or unit arrays since the +q -p terms cancel)
# SKL divergence, symmetric Kullback-Leibler divergence, symmetrical Kullback-Leibler divergence, symmetrized KL, etc.
#
JeffreysDivergence<-function(p,q)
{
d=length(p)
res=0
for(i in 1:d)
{res=res+(p[i]-q[i])*log(p[i]/q[i])}
res
}

normalizeHistogram<-function(h)
{
d=length(h)
result=rep(0,d)
wnorm=cumsum(h)[d]
for(i in 1:d)
{result[i]=h[i]/wnorm; }
result
}

#
# Weighted arithmetic center 
#
ArithmeticCenter<-function(set,weight)
{
dd=dim(set)
n=dd[1]
d=dd[2]
result=rep(0,d)
for(j in 1:d)
{
for(i in 1:n)
{result[j]=result[j]+weight[i]*set[i,j]}
}
result
}

#
# Weighted geometric center
#
GeometricCenter<-function(set,weight)
{
dd=dim(set)
n=dd[1]
d=dd[2]
result=rep(1,d)
for(j in 1:d)
{
for(i in 1:n)
{result[j]=result[j]*(set[i,j]^(weight[i]))}
}
result
}

cumulativeSum<-function(h)
{
d=length(h)
cumsum(h)[d]
}

#
# Optimal Jeffreys (positive array) centroid expressed using Lambert W0 function 
#
JeffreysPositiveCentroid<-function(set,weight) 
{
dd=dim(set)
d=dd[2]
result=rep(0,d)
a=ArithmeticCenter(set,weight)
g=GeometricCenter(set,weight)
for(i in 1:d)
	{result[i]=a[i]/LambertW0(exp(1)*a[i]/g[i])}
result
}

#
# A 1/w approximation to the optimal Jeffreys frequency centroid (the 1/w is a very loose upper bound)
#
JeffreysFrequencyCentroidApproximation<-function(set,weight) 
{
result=JeffreysPositiveCentroid(set,weight) 
w=cumulativeSum(result)
result=result/w
result
}

#
# By introducing the Lagrangian function enforcing the frequency histogram constraint
#
cumulativeCenterSum<-function(na,ng,lambda)
{
d=length(na)
cumul=0
for(i in 1:d)
{cumul=cumul+(na[i]/LambertW0(exp(lambda+1)*na[i]/ng[i]))}
cumul
}

#
# Once lambda is known, we get the Jeffreys frequency cetroid
#
JeffreysFrequencyCenter<-function(na,ng,lambda)
{
d=length(na)
result=rep(0,d)
for(i in 1:d)
{result[i]=na[i]/LambertW0(exp(lambda+1)*na[i]/ng[i])}
result
}

#
# Approximating lambda up to some numerical precision for computing the Jeffreys frequency centroid 
#
JeffreysFrequencyCentroid<-function(set,weight) 
{
na=normalizeHistogram(ArithmeticCenter(set,weight))
ng=normalizeHistogram(GeometricCenter(set,weight))
d=length(na)
bound=rep(0,d)
for(i in 1:d) {bound[i]=na[i]+log(ng[i])}
lmin=max(bound)-1 
lmax=0
while(lmax-lmin>1.0e-10)
{
lambda=(lmax+lmin)/2
cs=cumulativeCenterSum(na,ng,lambda);
if (cs>1)
{lmin=lambda} else {lmax=lambda}
}

JeffreysFrequencyCenter(na,ng,lambda)
}




demoJeffreysCentroid<-function(){
n<-2
d<-256
weight=c(0.5,0.5)
set<- array(0, dim=c(n,d))
set[1,]=lenah[,1]
set[2,]=barbarah[,1]
approxP=JeffreysPositiveCentroid(set,weight)
approxF=JeffreysFrequencyCentroidApproximation(set,weight)
approxFG=JeffreysFrequencyCentroid(set,weight) 

pdf(file="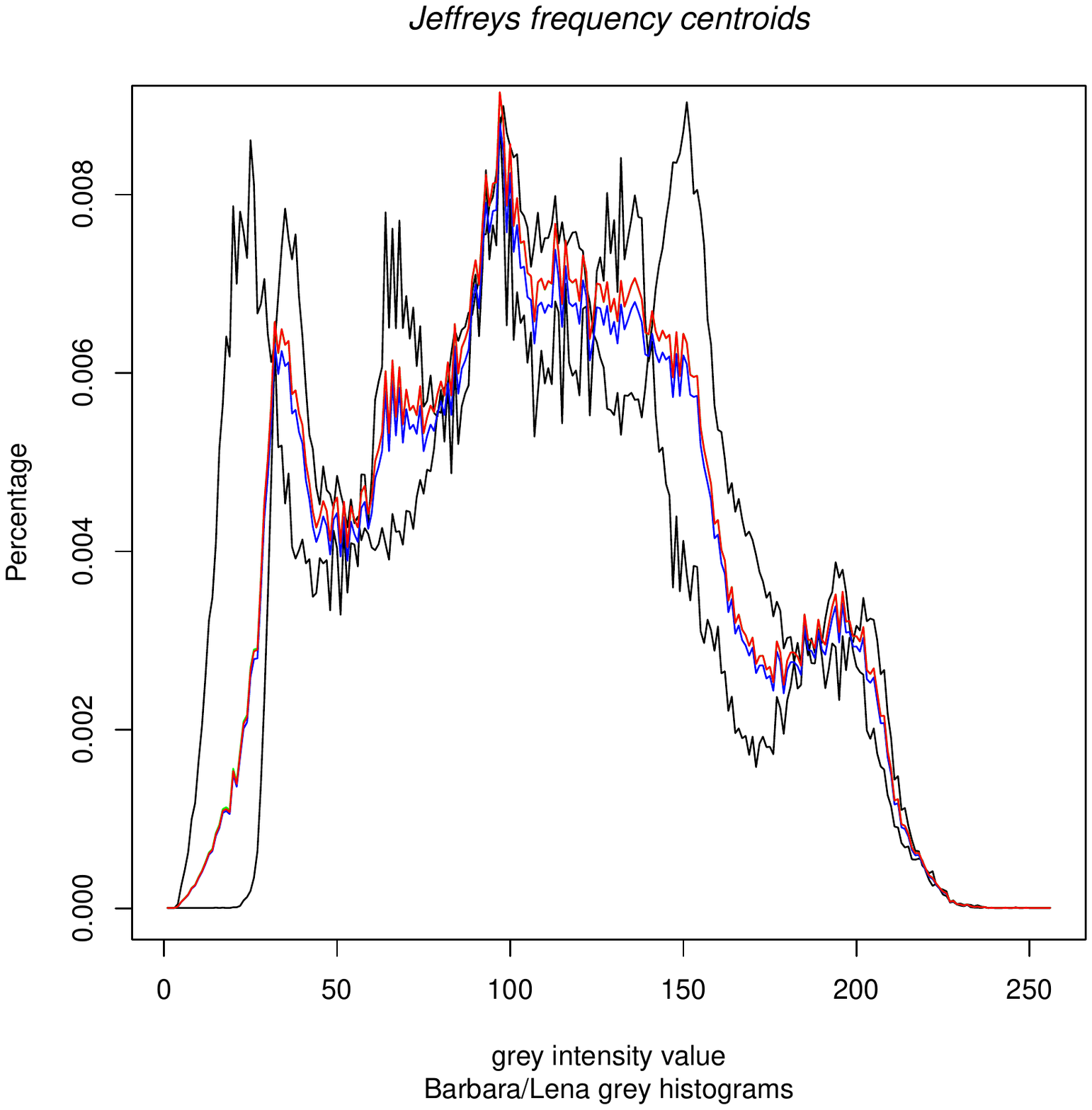", paper ="A4")
typedraw="l"
plot(set[1,], type=typedraw, col = "black", xlab="grey intensity value", ylab="Percentage")
lines(set[2,], type=typedraw,col = "black" )
lines(approxP,type=typedraw, col="blue")
# green not visible, almost coincide with red
lines(approxF, type=typedraw,col = "green")
lines(approxFG,type=typedraw, col="red")
title(main="Jeffreys frequency centroids", sub="Barbara/Lena grey histograms" , col.main="black", font.main=3)
dev.off()
graphics.off()
}

demoJeffreysCentroid()



randomUnitHistogram<-function(d)
{
result=runif(d,0,1); 
result=result/(cumsum(result)[d]) 
result
}

n<-10  # number of frequency histograms
d<-25  # number of bins
set<- array(0, dim=c(n,d))
weight<-runif(n,0,1); 
cumulw<-cumsum(weight)[n]

# Draw a random set of frequency histograms
for (i in 1:n)
{set[i,]=randomUnitHistogram(d)
weight[i]=weight[i]/cumulw;
}

A=ArithmeticCenter(set,weight)
G=GeometricCenter(set,weight)
nA=normalizeHistogram(A) # already normalized 
nG=normalizeHistogram(G)
# Simple approximation (half of normalized geometric and arithmetic mean)
nApproxAG=0.5*(nA+nG)

optimalP=JeffreysPositiveCentroid(set,weight) 
approxF=JeffreysFrequencyCentroidApproximation(set,weight)
# guaranteed approximation factor
AFapprox=1/cumulativeSum(optimalP);

# up to machine precision
approxFG=JeffreysFrequencyCentroid(set,weight) 


pdf(file="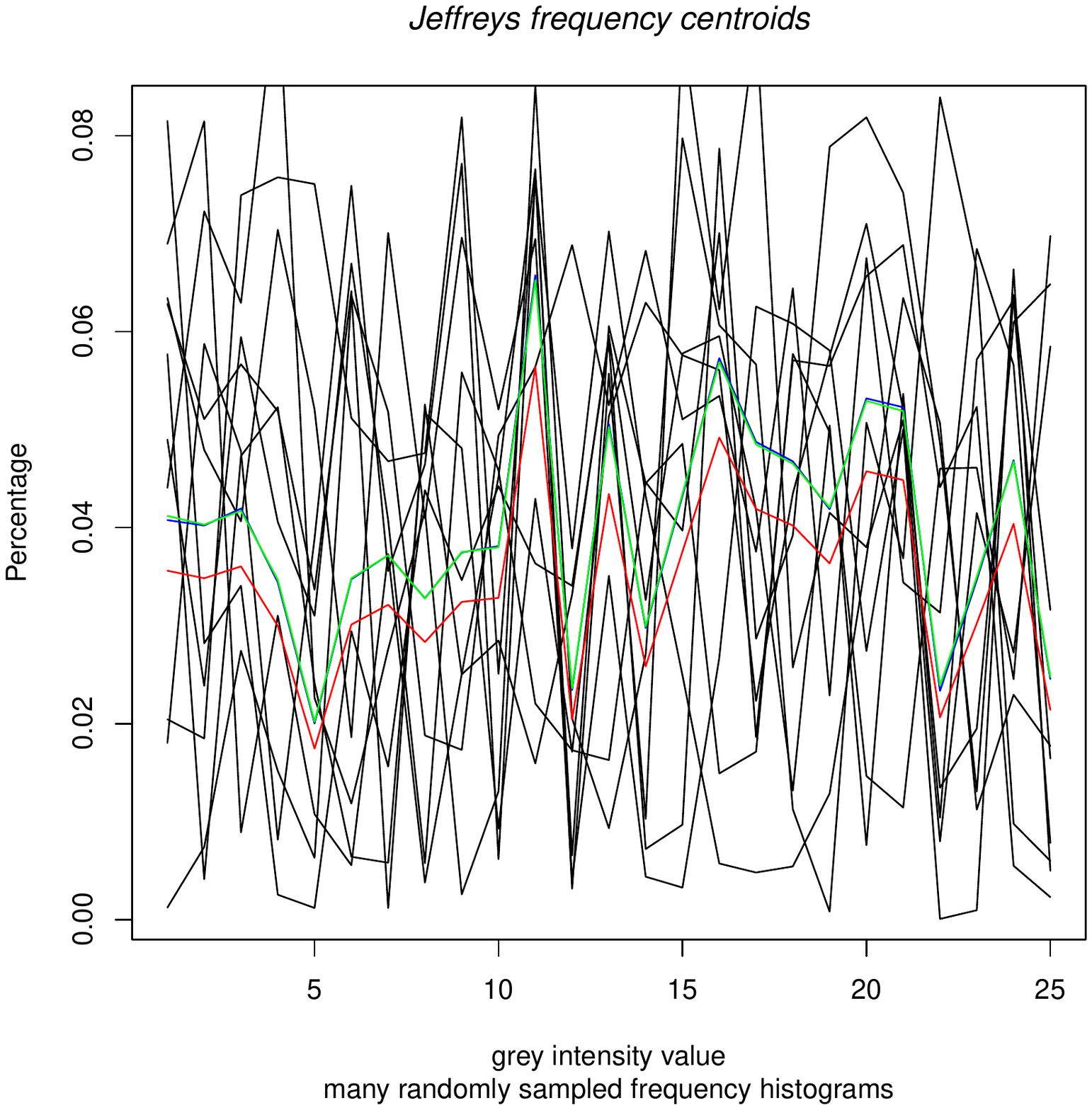", paper ="A4")
typedraw="l"
plot(set[1,], type=typedraw, col = "black", xlab="grey intensity value", ylab="Percentage")
for(i in 2:n) {lines(set[i,], type=typedraw,col = "black" )}
lines(approxFG,type=typedraw, col="blue")
lines(approxF, type=typedraw,col = "green")
lines(optimalP,type=typedraw, col="red")
title(main="Jeffreys frequency centroids", sub="many randomly sampled frequency histograms" , col.main="black", font.main=3)
dev.off()

# test the fixed-point iteration algorithm
Lambda<-function(set,weight) 
{
na=normalizeHistogram(ArithmeticCenter(set,weight))
ng=normalizeHistogram(GeometricCenter(set,weight))
d=length(na)
bound=rep(0,d)
for(i in 1:d) {bound[i]=na[i]+log(ng[i])}
lmin=max(bound)-1 
lmax=0
while(lmax-lmin>1.0e-10)
{
lambda=(lmax+lmin)/2
cs=cumulativeCenterSum(na,ng,lambda);
if (cs>1)
{lmin=lambda} else {lmax=lambda}
}

lambda
}

JeffreysFrequencyCentroid<-function(set,weight) 
{
na=normalizeHistogram(ArithmeticCenter(set,weight))
ng=normalizeHistogram(GeometricCenter(set,weight))
lambda=Lambda(set,weight)
JeffreysFrequencyCenter(na,ng,lambda)
}


KullbackLeibler<-function(p,q)
{
d=length(p)
res=0
for(i in 1:d)
{res=res+p[i]*log(p[i]/q[i])+q[i]-p[i]}
res
}


lambda=Lambda(set,weight)
ct=JeffreysFrequencyCentroid(set,weight)
na=normalizeHistogram(ArithmeticCenter(set,weight))
ng=normalizeHistogram(GeometricCenter(set,weight))

# Check that $\lambda=-\KL(\tilde{c}:\tilde{g})$
lambdap=-KullbackLeibler(ct,ng)
lambda
# difference
abs(lambdap-lambda)

OneIteration<-function()
{
lambda=-KullbackLeibler(c,ng)
c=JeffreysFrequencyCenter(na,ng,lambda)
}

NbIteration<-function(kk)
{
for(i in 1:kk)
{OneIteration()}
lambda
}

Lambda<-function(set,weight) 
{
na=normalizeHistogram(ArithmeticCenter(set,weight))
ng=normalizeHistogram(GeometricCenter(set,weight))
c=na
}

# Fixed point solution for the frequency centroid

Same<-function(p,q)
{
res=TRUE
d=length(p)

for(i in 1:d)
{if (p[i]!=q[i]) res=FALSE}

res
}

# Fixed point iteration
FixedPointSKL<-function(set,weight)
{
nbiter<<-0
na=normalizeHistogram(ArithmeticCenter(set,weight))
ng=normalizeHistogram(GeometricCenter(set,weight))
center=na

repeat{
lambda=-KullbackLeibler(center,ng)
nbiter<<-nbiter+1
nc=JeffreysFrequencyCenter(na,ng,lambda)

  if(Same(center,nc)==TRUE | (nbiter>100) ){
    break
  }
else
{center=nc}
}

center
}


#### comparison of the two methods
lambda=FixedPointSKL(set,weight)
fpc=JeffreysFrequencyCenter(na,ng,lambda)
cc=JeffreysFrequencyCentroid(set,weight) 
cumulativeSum(cc)
cumulativeSum(fpc)



TestJeffreysFrequencyMethods<-function()
{
n<-10  # number of frequency histograms
d<-25  # number of bins
set<- array(0, dim=c(n,d))
weight<-runif(n,0,1); 
cumulw<-cumsum(weight)[n]

# Draw a random set of frequency histograms
for (i in 1:n)
{set[i,]=randomUnitHistogram(d)
weight[i]=weight[i]/cumulw;}

J1=JeffreysFrequencyCentroid(set,weight)
J2=FixedPointSKL(set,weight)
cat("Bisection search final w=",cumulativeSum(J1)," Fixed-point iteration final w=",cumulativeSum(J2)," Number of iterations:",nbiter,"\n")
} 

TestJeffreysFrequencyMethods()
# That is all.

\end{lstlisting}

\end{document}